\documentclass[11pt,a4paper,twoside]{article}

\usepackage{a4wide, amssymb, amsmath, amsthm, graphics, comment, xspace, enumerate,enumitem}
\usepackage{graphicx,multirow,rotating,cite}
\usepackage[a4paper,colorlinks=true,citecolor=blue,urlcolor=blue,linkcolor=blue,bookmarksopen=true,unicode=true,
          pdffitwindow=true]{hyperref}
\usepackage[noend]{algpseudocode}
\usepackage{mathrsfs,epsfig,amsfonts,amscd,caption, tikz,breqn}
\usepackage{graphicx,multirow,rotating}

\usepackage{float}
\usepackage{xspace}
\usepackage[section]{algorithm}
\usepackage[noend]{algpseudocode}
\usepackage{mathrsfs,epsfig,amsfonts,amscd,caption, tikz,breqn}
\newtheorem{te}{Theorem}[section]

\newtheorem{co}{Corollary}[section]
\newtheorem{lemma}{Lemma}[section]

\newtheorem{conjecture}{Conjecture}[section]

\newcommand{\beq}{\begin{eqnarray}}
\newcommand{\eeq}{\end{eqnarray}}
\newcommand{\beqs}{\begin{eqnarray*}}
\newcommand{\eeqs}{\end{eqnarray*}}

\allowdisplaybreaks

\newcommand{\irr}{\rm irr}

\begin{document}
 \title{{Non-regular graphs with minimal total irregularity} }
 
 \date{}

\maketitle
\begin{center}
{\large \bf  Hosam Abdo, Darko Dimitrov}
\end{center}
\baselineskip=0.20in
\begin{center}
{\it Institute of Computer Science, Freie Universit\"{a}t Berlin,
\\ Takustra{\ss}e 9, D--14195 Berlin, Germany} \\E-mail: {\tt [abdo,darko]@mi.fu-berlin.de} \\[2mm]
\end{center}
\vspace{6mm}
\begin{abstract}
The {\it total irregularity} of a simple undirected graph $G$ is defined
as
$\irr_t(G) =$ $\frac{1}{2}\sum_{u,v \in V(G)}$ $\left| d_G(u)-d_G(v) \right|$,
where $d_G(u)$ denotes the degree of a vertex $u \in V(G)$.
Obviously, $\irr_t(G)=0$ if and only if $G$ is regular.
Here, we characterize the non-regular graphs with minimal total irregularity and
thereby resolve the recent conjecture by 
Zhu, You and Yang~\cite{zyy-mtig-2014} 
about the lower bound on the minimal total irregularity of non-regular connected graphs.
We show that the conjectured lower bound of $2n-4$ is attained only if non-regular connected graphs of even order are considered,
while the sharp lower bound of $n-1$ is attained by graphs of odd order.
We also characterize the non-regular graphs with the second and the third smallest total irregularity.
\end{abstract}
{\small \hspace{0.25cm} \textbf{Keywords:} the total irregularity, extremal graphs
%
\section[Introduction]{Introduction}
%
\medskip
All graphs considered here are undirected and simple (i.e., loops and multiple edges are not allowed).
Let $G$ be a graph of order $n=|V(G)|$ and size $m=|E(G)|$.
For $v \in V(G)$, the degree of $v$, denoted by $d_G(v)$, is the number of edges incident
to $v$.
$G$ is {\em regular} if all its
vertices have the same degree, otherwise it is {\em irregular}.
There have been proposed many approaches, including those in
\cite{Alavi-87,  Alavi-88, Albertson, Bell:1,  Char-87, Char-in-88, Char-88, CollSin-57}, 
that characterize how irregular a given graph is.
In this paper, we focus on the so-called {\it total irregularity}  of a graph \cite{Dimit-Abdo2}, defined as
\beq \label{eqn:003-t}
{\irr_t} (G) = \frac{1}{2} \sum_{u, v \in V(G)} |d_G(u)-d_G(v)|.
\eeq
The total irregularity is related to the {\it irregularity} of a graph, defined as
$
{\irr(G)} = \sum_{uv\in E(G)}$ $|d_G(u)-d_G(v)|.
$
The latter measurement was introduced by Albertson~\cite{Albertson} and investigated in several works
including \cite{Dimit-Abdo1, HanM05, HenRaut-07}. 
For motivations of introducing the total irregularity as new irregularity measure, we refer an interested reader to \cite{Dimit-Abdo2}.
Both measures, the irregularity of a graph and the total irregularity of a graph, depend only on one single parameter, 
namely the pairwise difference of vertex degrees. 
A comparison of  $\irr$ and $\irr_t$ was considered in \cite{DS-citig-2012}.
There, it  was shown that
$ \irr_t(G) \leq n^2 \, \irr(G)/4$ and
when $G$ is a tree, then 
$\irr_t (G) \leq (n-2) \, \irr(G)$.
Also, it was shown that among all trees of the same order, the star has the maximal total irregularity.

In \cite{Dimit-Abdo2}, graphs
with maximal  total irregularity were fully characterized and the upper bound on the total irregularity of a graph was presented
(see Figure~\ref{f:H_n} and Corollary~\ref{thm-graphs-max-irr_t}).
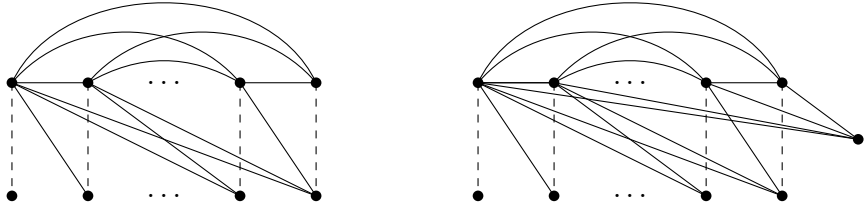
\begin{figure}[htbp] 
\vspace{-0.5cm}
\centering
\vspace{-0.1cm}
\begin{minipage}[h]{6cm}
\begin{tikzpicture}
        \path (0,0) coordinate (P0);
        \path (1,0) coordinate (P1);
        \path (3,0) coordinate (P2);
        \path (4,0) coordinate (P3);
        \path (0,-1.5) coordinate (Q0);
        \path (1,-1.5) coordinate (Q1);
        \path (3,-1.5) coordinate (Q2);
        \path (4,-1.5) coordinate (Q3);

        \foreach \x in {0,...,3} {
	\fill (P\x) circle (2pt);
	\fill (Q\x) circle (2pt);}
        \foreach \x in {0,...,3} {
             \draw[dashed] (P\x) -- (Q\x) ;}

\path   (2,0) node {$\ldots$};
\path   (2,-1.5) node {$\ldots$};

             \draw (P0) -- (Q1) ;
             \draw (P0) -- (Q2) ;
             \draw (P0) -- (Q3) ;
             \draw (P1) -- (Q2) ;
             \draw (P1) -- (Q3) ;
             \draw (P2) -- (Q3) ;
             \draw (P1) -- (P0) ;
             \draw (P3) -- (P2) ;

             \draw (P1) to[out=30,in=150] (P2) ;
             \draw (P3) to[out=115,in=65] (P0) ;
             \draw (P3) to[out=130,in=50] (P1) ;
             \draw (P2) to[out=130,in=50] (P0) ;
\end{tikzpicture}
\end{minipage}
\begin{minipage}[h]{6cm}
\begin{tikzpicture}

        \path (0,0) coordinate (P0);
        \path (1,0) coordinate (P1);
        \path (3,0) coordinate (P2);
        \path (4,0) coordinate (P3);
        \path (0,-1.5) coordinate (Q0);
        \path (1,-1.5) coordinate (Q1);
        \path (3,-1.5) coordinate (Q2);
        \path (4,-1.5) coordinate (Q3);
        \path (5,-0.75) coordinate (R);

        \foreach \x in {0,...,3} {
	\fill (P\x) circle (2pt);
	\fill (Q\x) circle (2pt);}
	\fill (R) circle (2pt);
        \foreach \x in {0,...,3} {
             \draw[dashed] (P\x) -- (Q\x) ;}

\path   (2,0) node {$\ldots$};
\path   (2,-1.5) node {$\ldots$};

             \draw (P0) -- (Q1) ;
             \draw (P0) -- (Q2) ;
             \draw (P0) -- (Q3) ;
             \draw (P1) -- (Q2) ;
             \draw (P1) -- (Q3) ;
             \draw (P2) -- (Q3) ;
             \draw (P1) -- (P0) ;
             \draw (P1) -- (P0) ;
             \draw (P3) -- (P2) ;
             \draw (P0) -- (R) ;
             \draw (P1) -- (R) ;
             \draw (P2) -- (R) ;
             \draw (P3) -- (R) ;
        
             \draw (P3) to[out=115,in=65] (P0) ;
              \draw (P1) to[out=30,in=150] (P2) ;
             \draw (P3) to[out=130,in=50] (P1) ;
             \draw (P2) to[out=130,in=50] (P0) ;
\end{tikzpicture}
\end{minipage}
\vspace{-0.1cm}
\caption{Graphs with maximal total irregularity $H_n$ (with dashed edges) and $\overline{H}_n$
               (without dashed edges) for even and odd $n$, respectively.\label{max-total-irr}} \label{f:H_n}
\end{figure}
\vspace{-0.2cm}%
\begin{co}[\cite{Dimit-Abdo2}] \label{thm-graphs-max-irr_t}
For a graph $G$ with $n$ vertices, it holds that 
\beq 
\irr_t(G) \leq 
    \begin{cases}  
	   \frac{1}{12}(2n^3 - 3n^2 - 2n)    &\mbox{n even,}\\
                 \\
	    \frac{1}{12}(2n^3 - 3n^2 - 2n +3)  &\mbox{n odd.} \nonumber
	\end{cases}
\eeq
Moreover, the bounds are sharp.
\end{co}
\vspace{-0.2cm}
In \cite{ljz-mtiug-2014, ljz-mtibg-2014}, the unicyclic and bicyclic graphs, respectively, with maximal total irregularity were determined.

The lower bound on the total irregularity of general graphs is trivial, since
it is obvious that  the total irregularity of a graph is zero if and only if the graph is regular.
Also by definition the total irregularity is nonnegative.
However it is not trivial to determine the lower bounds on the total irregularity
of special classes of graphs and the total irregularity
of non-regular graphs.
In~\cite{zyy-mtig-2014},  Zhu, You and Yang investigated 
the minimal total irregularity of the connected graphs, determined the minimal, the
second minimal, the third minimal total irregularity of trees, unicyclic graphs, bicyclic
graphs. 
They also proposed the following conjecture.

\begin{conjecture}[\cite{zyy-mtig-2014}] \label{conj-graphs-min-irr_t}
Let $G$ be a simple connected graph with $n$ vertices. If $G$ is a non-regular
graph, then $irr_t(G) \geq 2n-4$.
\end{conjecture}
\noindent
In the next section, we characterize the non-regular graphs with minimal total irregularity and
thereby resolve the above conjecture.
We show that Conjecture~\ref{conj-graphs-min-irr_t} is true only for non-regular connected graphs of even order,
while the actual sharp lower bound of $n-1$ is achieved by graphs of odd order.

By $D(G)$ we denote the set of the vertex degrees of a graph $G$, i.e., $D(G)=\{d(v)\,|\,v\in V\}$.
Given an undirected graph, a {\it degree sequence} is a monotonic non-increasing 
sequence of the degrees of its vertices.
A {\it graphical sequence} is a sequence of numbers which can be the degree sequence of some graph.
In general, several graphs may have the same graphical sequence.
In order to show that a given sequence of non-negative integers  is graphical,
one may use the following characterization by Erd\H{o}s and Gallai.
\begin{te}[\cite{eg-ggdv-60}]\label{te-erdos-gallai}
A sequence $d_1 \geq d_2 \geq \dots \geq d_n$ of non-negative integers
with even sum is graphical if and only if
\beq \label{eq-thm-50}
\sum\limits_{i=1}^{r} d_i \leq r(r-1) + \sum\limits_{i=r+1}^{n} \min(r,d_i),
\eeq 
for all $1 \leq r \leq n-1$.
\end{te}
Tripathi and Vijay \cite{tv-nteg-03} showed that the inequality (\ref{eq-thm-50}) need be checked 
only for as many $r$ as there are distinct terms in the sequence, not for all $1 \leq r \leq n-1$.
Denote the indices $1 \leq i \leq n-1$, with $d_i > d_{i+1}$, by
$\sigma_1, \sigma_2, \dots, \sigma_l$, and define $\sigma_l=n$.
\begin{te}[ \cite{tv-nteg-03}]\label{te-tripathi-vijay}
In Theorem~\ref{te-erdos-gallai} it suffices to check the inequalities (\ref{eq-thm-50})
for $r= \sigma_1, \sigma_2, \dots, \sigma_l$.
\end{te}

\section[Results]{Results}

The only non-regular connected graph of order at most three is the path with three vertices,
whose total irregularity is $2n-4=2$. Therefore, in the sequel we consider connected graphs of order at least four.
First, we present two results that will be used to obtain the main results later in this section.

\begin{lemma} \label{pro-DS-caridnality}
Let $G$ be a connected graph of order $n >3$ with $|D(G)| \geq 3$.
Then, there exists a connected graph $H$ of same order as $G$ with $|D(H)| = 2$, 
such that each degree in $H$ occurs at least two times  and ${\irr_t} (H) < {\irr_t} (G)$.
\end{lemma}
\begin{proof}
Assume that the claim of the proposition is false, i.e., $D(G)=\{d_1, d_2, \dots, d_k\}$, $3 \leq k \leq n-1$.
Also, we assume that  $d_1 > d_2 > \dots >d_k$.
Let $D_G=(d_G(v_1), d_G(v_2), \dots, d_G(v_n))$ be the degree sequence of $G$ and
$p$ be the smallest index such that $d_G(v_p)=d_3$.
We apply a set of transformations to $D_G$ obtaining a sequence $D_H=(d_H(v_1), d_H(v_2), \dots,$ $ d_H(v_n))$ and consider the  difference
\beq \label{diff_G-H}
\sum_{u, v \in V(G)} \left(|d_G(u)-d_G(v)|   -  |d_H(u)-d_H(v)|\right).
\eeq
We distinguish two cases regarding $p$.

\smallskip
\noindent
{\bf Case $1$.}  $p<n.$

\noindent
We apply the following assignments to $D_G$:
$$d_H(v_i):=d_3+1,  \, i=1, \dots p-1,  \quad \text{and}   \quad d_H(v_i):=d_3,  \, i=p+1, \dots n.$$ 

\noindent
After these assignments, for all pairs of vertices $v_i$ and $v_j$, with $d_G(v_i) \neq d_G(v_j) \neq d_G(v_p)$, it holds that
\beq \label{exp-00}
 |d_G(v_i)-d_G(v_j)| -  |d_H(v_i)-d_H(v_j)| \geq 1.  \nonumber
\eeq
For the rest of the pairs of vertices $v_i$ and $v_j$, it holds that
\beq \label{exp-001}
 |d_G(v_j)-d_G(v_j)| -  |d_H(v_i)-d_H(v_j)| \geq 0.  \nonumber
\eeq
Thus, it follows that the difference (\ref{diff_G-H}) is positive.

\noindent
If $n-p+1$ and $d_3$ are odd, the sum of the elements of $D_H$ is also odd, and $D_H$ cannot be graphical sequence,
since the parity condition of Theorem~\ref{te-erdos-gallai} is not satisfied. 
In this case we apply additional assignments such that $d_3$ occurs even times in the sequence $D_H$.
We distinguish two cases:
\begin{itemize}
\item If $p-3 < n-p$, then  $d_H(v_{p-1}):=d_3$. 

\noindent
In this case,   $p-2$ summands  in $\sum_{d_H(u), d_H(v) \in D_H} |d_H(u)-d_H(v)|$ increase by one,
$n-p+1$ summands decrease by one, and the rest remain unchanged. Thus, the total change 
$p-2- (n-p+1)$ is negative.

\item If $p-3 \geq n-p$, then  $d_H(v_{p}):=d_3+1$. 

\noindent
After this assignment  $p-1$ summands  in $\sum_{d_H(u), d_H(v) \in D_H} |d_H(u)-d_H(v)|$ decrease by one,
$n-p$ summands increase by one, and the rest remain unchanged. 
Here also, the total change $-(p-1)+n-p$ is negative.
\end{itemize}

\smallskip
\noindent
{\bf Case $2$.}  $p=n.$

\noindent
In this case $D_G$ is comprised of three degrees and
the degree $d_3$ occurs once. Assume that $d_1$ occurs $x$ times, $n-2 \geq x \geq 1$. Then,  $d_2$ occurs $n-x-1$ times. We, perform the following assignments:
$$
d_H(v_{p-1}):=d_3 \quad \text{and}   \quad d_H(v_{i}):=d_3+1, i=1, \dots p-2. 
$$

\noindent
For every $v \in V(G)$ consider the pair $(d_G(v), d_H(v))$.
After the above transformation there are  $x$ pairs $(d_1, d_3+1)$,
$n-x-2$ pairs $(d_2, d_3+1)$, one pair $(d_2, d_3)$ and one pair $(d_3, d_3)$.
It holds that
\begin{align}
\hspace{10mm} &\sum_{ \substack{  i=1, \dots, x \\j=x+1, \dots, n-2} } \left(|d_G(v_i)-d_G(v_j)|   -  |d_H(v_i)-d_H(v_j)|\right)= x(n-x-2)(d_1-d_2 ), \nonumber
\end{align}
\begin{align}
\hspace{7mm} &\sum_{ i=1, \dots, x } \left(|d_G(v_i)-d_G(v_{n-1})|   -  |d_H(v_i)-d_H(v_{n-1})|\right)= x( d_1 - d_2-1),  \nonumber
\end{align}

\begin{align}
 &\sum_{ i=1, \dots, x } \left(|d_G(v_i)-d_G(v_{n})|   -  |d_H(v_i)-d_H(v_{n})|\right)= x( d_1 - d_3 -1),  \nonumber
\end{align}
\begin{align}
\hspace{6mm} &\sum_{ i=x+1, \dots, n-2 } \left(|d_G(v_i)-d_G(v_{n-1})|   -  |d_H(v_i)-d_H(v_{n-1})|\right)=  -(n-x-2),  \nonumber
\end{align}
\begin{align}
\hspace{17mm} &\sum_{ i=x+1, \dots, n-2 } \left(|d_G(v_i)-d_G(v_{n})|   -  |d_H(v_i)-d_H(v_{n})|\right)= (n-x-2)(d_2-d_3-1)),  \nonumber
\end{align}
and
\begin{align}
\hspace{-5mm} &|d_G(v_{n-1})-d_G(v_{n})|   -  |d_H(v_{n-1})-d_H(v_{n})|=d_2-d_3.  \nonumber
\end{align}

\noindent
Thus, the difference  (\ref{diff_G-H}) is
\beq \label{exp-070}
  && x(n-x-2)(d_1-d_2) + x(d_1-d_2-1) +  x(d_1-d_3-1) -  (n-x-2) +   \nonumber \\
  && (n-x-2)(d_2-d_3-1) +d_2-d_3. 
\eeq
Since $ d_1- d_2\geq 1$, $d_1-d_2-1 \geq 0$, $d_1-d_3-1 \geq 1$,  $d_2-d_3-1 \geq 0$, $d_2-d_3 \geq 1$ and $x \leq n-2$,
the lower bound on (\ref{exp-070}) is 
\beq \label{exp-075}
&& x(n-x-2)  +  x -  (n-x-2)  +1=(x-1)(n-x-2)  + x +1.
\eeq
From $x \geq 1$, it follows that (\ref{exp-075}), and therefore (\ref{exp-070}) and
(\ref{diff_G-H}) are positive.

\noindent
Observe that the sequence $D_H$, in both Cases $1$ and $2$ is comprised of degrees $d_3+1$ and $d_3$, 
where $d_3$  occurs even times and $d_3 \leq n-3$. 
Also, note that $D_H$ does not  necessarily satisfy the parity condition of Theorem~\ref{te-erdos-gallai}.
For example, this is a case precisely when $n$ is odd and $d_3$ is even.  
In this case, we apply the following assignment:
$$
d_H(v_i):=d_H(v_i)+1,\quad i=1, \dots n.
$$
Thus, now we have a degree sequence 
$D_H$ with $2 \leq y \leq n-2$ occurrences of degree $d+1$,
and $n-y$ occurrences of degree $d$, where $d$ is either $d_3$ or  $d_3+1$.

Next we show that $D_H$ is a degree sequence of a graph $H$, by showing that $D_H$ satisfies the condition (\ref{eq-thm-50}).
By Theorem~\ref{te-tripathi-vijay} it suffices to show that (\ref{eq-thm-50}) is satisfied for $r=y$.
With respect to $y$, we consider two cases.

\begin{itemize}
\item $y \leq d+1$. Then, for $r=x$, (\ref{eq-thm-50}) can be written as
\beq \label{exp-0610}
y(d +1) \leq y(y-1) +y(n-y),  \quad \text{or} \quad y(d +1) \leq y(n-1), \nonumber
\eeq
which holds since $d \leq d_3+2$ and $d_3 \leq n-3$.
\item $y > d+1$. In this case, (\ref{eq-thm-50}) can be written as
\beq \label{exp-0610}
y(d +1) \leq y(y-1) +(n-y)(d+1),  \quad \text{or} \quad 0 \leq y(y-d-2) +(n-y)(d+1), \nonumber
\eeq
and it is satisfied because $d +2 \leq y$ and $y < n$.
\end{itemize}

\noindent
Thus, we have shown that $D_H$ is a graphical sequence of a graph $H$, and
$$
{\irr_t} (G) = \frac{1}{2}\sum_{u, v \in V(G)} |d_G(u)-d_G(v)| \quad  >  \quad  \frac{1}{2} \sum_{u, v \in V(H)} |d_H(u)-d_H(v)| = {\irr_t} (H),
$$
which is a contradiction to the initial assumption that $G$ is a non-regular graph with minimal total irregularity.
\end{proof}

\noindent
Lemma~\ref{pro-DS-caridnality}  shows that a non-regular graph with minimal total irregularity 
must have degree set of cardinality  two and both degrees may  occur more than once. 
Obviously, the total irregularity is smaller if the
two degrees differ  as little as possible, i.e., if they differ by one.
In the next lemma we present sharper conditions on a non-regular graph of odd degree with minimal total irregularity.

\begin{lemma} \label{pro-DS-caridnality-2}
Let $G$ be a connected graph of odd order $n >3$ with  $|D(G)|=2$, such that 
each degree occurs at least two times.
Then, there exists a connected graph $H$ of same order $n$ with degree set $D(H)=\{d+1, d\}$, such that one of the 
degrees occurs only once and ${\irr_t} (H) < {\irr_t} (G)$.
Moreover, $H$ has one of the following  degree sequences
\beq \label{exp-30}
&& (d+1, d, \dots, d, d), \quad 1 \leq d \leq n-2, \quad \text{or}  \nonumber \\
&& (d+2, d+2, \dots, d+2, d+1), \quad 1 \leq d \leq n-4. 
\eeq
\end{lemma}
\begin{proof}
All possible sequences, where one degree occurs only ones and 
the difference between the two degrees is one,
 are those from (\ref{exp-30}). In the sequel,  we will show that they are indeed graphical sequences.
Since $n$ is odd, $d$ must be odd as well, because otherwise the parity condition 
of Theorem~\ref{te-erdos-gallai} will be not satisfied.
Observe that the range of values of $d$ in (\ref{exp-30}) follow from the fact that $n$ and $d$ are odd and
$d \leq n-1$.

\noindent
First, we show that, for a fixed 
$d$ with $1 \leq d \leq n-2$, the sequence $(d+1, d, \dots, d, d)$ is graphical.
For that, we need to show in addition that (\ref{eq-thm-50}) holds.
In this case, by Theorem~\ref{te-tripathi-vijay}, it suffices to show that (\ref{eq-thm-50}) is satisfied for $r=1$.
Then,  (\ref{eq-thm-50}) can be written as
$$ d+1 \leq \sum\limits_{i=2}^{n} \min(r,d),$$
which obviously holds since $d \leq n-2$ and $r=1$.

\smallskip
\noindent
Next, we show that, for a fixed 
$d$ with $1 \leq d \leq n-4$, the sequence $(d+2, d+2, \dots, d+2, d+1)$ is graphical.
Since $(n-1)(d+2) + d+1$ is even, the parity condition of
Theorem~\ref{te-erdos-gallai} is satisfied. 
By Theorem~\ref{te-tripathi-vijay}, it suffices to show that (\ref{eq-thm-50}) is satisfied for $r=n-1$.
Then,  (\ref{eq-thm-50}) can be written as
$$(n-1)(d+2) + d+1 \leq (n-1)(n-2) + d+1,  \quad \text{or} \quad 0 \leq (n-1)(n-d-4) + d +1.$$
The  last expressions holds since  $d \leq n-4$.

\smallskip
\noindent
Next, we show that ${\irr_t} (H) < {\irr_t} (G)$.
Assume that $G$ has $y$ vertices of degree $d$ and 
$n-y$ vertex of degree $d+1$, where $2 \leq y<n$. Then,
\beq \label{exp-50}
{\irr_t} (G)=y(n-y).
\eeq
With the given constrains, ${\irr_t} (G)$ reaches its minimum of $2n-4$ for $y=2$ and $y=n-2$,
which is larger than ${\irr_t} (H)=n-1$, for $n >3$.
\end{proof}

For an illustration of the above lemma consider the degree sequence
$(n-1, n-2, \dots, n-2, n-2).$ 
A  graph with this degree sequence
can be constructed by deleting $\lfloor n/2 \rfloor$ edges  from
$K_n$, such that no two deleted edges have a common endvertex. 

For the degree sequence
$(n-2, n-2, \dots, n-2, n-3),$ 
a corresponding  graph 
can be constructed by deleting $\lfloor n/2 \rfloor-1$ edges  from
$K_n$, such that no two deleted edges have a common endvertex. 
There remain $3$ vertices with degrees $n-1$. Finally, delete the two edges
that connect one of those vertices with the remaining two.

Now, we are ready to present the sharp lower bound on the total irregularity of the 
connected non-regular graphs, as well as their second and the third minimal value
regarding the total irregularity.

\begin{te}\label{te-graphs-min-irr_t-10}
Let $G$ be a connected non-regular graph with $n$ vertices. Then $irr_t(G) \geq n-1$.
Moreover, this bound is obtained by graphs of odd order, characterized in Lemma~\ref{pro-DS-caridnality-2}.
\end{te}

\begin{proof}
Let $H$ be a graph with the minimal total irregulariy.
By Lemmas~\ref{pro-DS-caridnality} and \ref{pro-DS-caridnality-2}, the  degree sequence of $H$ consists of
two deferent degrees that differ by one. Let denote them by $d$ and $d+1$.
Assume that $H$ has $y$ vertices of degree $d$ and 
$n-y$ vertex of degree $d+1$, where $1 \leq y<n$. 
Its total irregularity is $n(n-y)$ (as in the expression (\ref{exp-50})).
With the given constrains, ${\irr_t} (H)$ reaches its minimum of $n-1$ for $y=1$ and $y=n-1$.
If $n$ is even, we cannot obtain a graphical sequence for $y=1$ and $y=n-1$, since then, the sum of all degrees is odd. 
Thus  $y=1$ or $y=n-1$ are feasible solutions only when $n$ is odd.
In this case also the degree that occurs only once must be even.
In Lemma~\ref{pro-DS-caridnality-2}~ the degree sequences of those graphs were fully characterized and 
their total irregularities are $n-1$.
\end{proof}

\begin{te}\label{te-graphs-min-irr_t-20}
The second and the third smallest  value of the total irregularity of connected non-regular graphs of order $n$ are  $2n-4$ and $2n-2$,
respectively, and can be obtained by graphs of order $n$ with an arbitrary parity. 
\end{te}
\begin{proof}
The next smallest solutions of  (\ref{exp-50}) that correspond to a graphic sequence,
as it was shown by Lemma~\ref{pro-DS-caridnality},  are
 $y=2$ and $y=n-2$,  and they are feasible as well when $n$ is even. The corresponding graphs have 
 total irregularity of $2n-4$. 
 
A candidate that may have smaller total irregularity than the above solutions is
 a graph with two degrees such that they differ by more than one,  and one degree sequence 
 occurs only once.
 Obviously, if the difference between the two degrees is larger (than two) the total irregularity
 will be larger as well.
 Thus, let consider first the case where the degrees differ by two. 
 If it results in total irregularity that is larger than $2n-4$,
 then we do not need to check the cases where the difference between the two degrees is larger than two.
 Indeed, the total irregularity of a graph with degree sequence $(d, d-2, \dots, d-2, d-2),  4 \leq d$, if such graph exists,
 is $2n-2$.  Therefore, it follows that $2n-4$ is the second smallest value of the total irregularity
 of connected non-regular graphs.

Note that for every graph the total irregularity is an even number, since the number of vertices
of odd degree is even.
Thus, the value $2n-3$ is excluded as value of total irregularity and
the next candidate for the third smallest value is $2n-2$.
Next we show that the degree sequences
\beq \label{exp-30-10}
 (d, d-2, \dots, d-2, d-2), && \quad d=3 \; \text{and} \; n=4; \; \text{or}  \; 4 \leq d,  d  \; \text{is even} \; \text{and} \;d+1<n; \nonumber \\
&&  \quad  \text{or}  \; 5 \leq d,   d \; \text{is} \; \text{odd},  n \; \text{is even} \;  \text{and} \;d+1<n; \nonumber \\
 (d, d, \dots, d, d-2), &&\quad 4 \leq d, d \; \text{is} \; \text{even} \; \text{and} \;d+1<n; \nonumber \\
&&  \quad  \text{or}  \; 3 \leq d, d \; \text{is} \; \text{odd}, n \; \text{is even} \; \text{and} \; \;d+1<n, \nonumber 
\eeq
are graphical. The condition for the degree sequences ensures that the parity condition of Erd\H{o}s-Gallai theorem is satisfied,
and that degree sequence may belong to a connected graph.
Next, we show that the above degree sequences satisfy also the relation (\ref{eq-thm-50}).

\smallskip
\noindent
First consider the sequences $(d, d-2, \dots, d-2, d-2)$. 
By Theorem~\ref{te-tripathi-vijay}, it suffices to show that (\ref{eq-thm-50}) is satisfied for $r=1$.
Thus, (\ref{eq-thm-50}) can be written as
$$d \leq   \sum\limits_{i=2}^{n} \min(r,d-2).$$
Obviously, the inequality holds, since $r=1$ and $d<n-1$.

\smallskip
\noindent
Next, consider the sequences $(d, d, \dots, d, d-2)$. 
In this case, by Theorem~\ref{te-tripathi-vijay}, it suffices to show that (\ref{eq-thm-50}) is satisfied for $r=n-1$.
Then, (\ref{eq-thm-50}) can be written as
$$(n-1)d \leq  (n-2)(n-1)  +d-2,  \quad \text{or} \quad 0 \leq  (n-2-d)(n-1)  +d-2.$$
From $d < n-1$ and $n \geq 3$, it follows that the last inequalities hold.
\end{proof}

\noindent
We would like to note, that for disconnected non-regular graphs, similarly as for connected graphs, one may obtain
the above presented bounds simply by including zero as a possible
degree.
For example, the sharp lower bound in the case of disconnected non-regular graphs of $n-1$
can be obtained  by a graph of odd order $n$ with degree sequence
$(1, 1, \dots, 1, 0)$, consisting of  $(n-1)/2$ pairs of adjacent vertices and one isolated vertex.


%
\end{document}